\documentclass[12pt]{article}
\usepackage{a4wide}
\usepackage[T2A]{fontenc}
\usepackage[cp1251]{inputenc}
\usepackage[russian]{babel}
\usepackage{graphics}
\usepackage{amsmath, amsfonts, amssymb, amsthm,amscd}

\newcommand{\eps}{\varepsilon}
\newcommand{\ls}{\leqslant}
\newcommand{\gs}{\geqslant}

\newtheorem{state}{Утверждение}

\newtheorem{thm}{Теорема}
\newtheorem{dfn}{Определение}

\begin{document}

\title{Об улучшенной оценке меры кластерной структуры в компактном метрическом пространстве}
\author{Пушняков~A.\,C. \\
        \small{141700, Московская облаcть, г. Долгопрудный, Институтский пер., 9, МФТИ} \\
        \small{pushnyakovalex@mail.ru} \\
        \small{Поступила в редакцию 09.08.2017 г.}}
\date{}

\maketitle

\begin{abstract}

Рассматривается компактное метрическое пространство с ограниченной борелевской мерой. Под $r$"~кластером понимается любое измеримое множество диаметра не более $r$. Исследуется взаимосвязь мер кластерных структур~--- наборов попарно отделенных друг от друга кластеров,~--- и распределения расстояний. Показано, что полученная ранее нижняя оценка меры максимальной кластерной структуры является неулучшаемой в асимптотическом смысле. Предлагается ввести дополнительные ограничения на распределение расстояний, позволяющие улучшить оценку меры максимальной кластерной структуры.

\textit{Ключевые слова}: кластеризация, компактное метрическое пространство, борелевская мера, метрика Хаусдорфа, теорема Бляшке.

\end{abstract}

\section{Введение}

Мы рассматриваем компактное метрическое пространство с ограниченной борелевской мерой (компактная метрическая тройка Громова \cite{gromov2007metric, vershik1998univer}).
С точки зрения задачи кластеризации интересен случай, когда существует кластерная структура, мера которой близка к мере всего пространства.
В данном случае можно утверждать, что метрическое пространство представляется в виде объединения кластеров.

В предыдущей статье \cite{push17} предлагались следующие ограничения на распределение расстояний, соответствующие интуитивному представлению метрики, являющейся объединением кластеров.
Вначале все расстояния разделяются на \textit{короткие}, \textit{средние} и \textit{длинные} ребра в зависимости от параметра $r$, отвечающего за диаметр кластеров.
Короткие ребра соответствуют \textit{внутрикластерным} расстояниям, а длинные ребра~--- \textit{межкластерным} расстояниям соответственно.
Первое ограничение состоит в том, что мера средних ребер должна быть мала.
Второе ограничение обусловливается тем, что мы ищем кластерную структуру из ровно $k$ кластеров: среди любых $k + 1$ точки какие-то две должны попасть в один кластер, поэтому мы требуем, чтобы мера $k + 1$"~антиклик~--- наборов из $k + 1$ точки без коротких ребер,~--- также была мала. Оба ограничения формулируются параметрически.

В описанных выше ограничения была получена нижняя оценка на меру кластерной структуры максимальной меры, которая стремится к мере всего пространства (при стремлении соответствующих параметров к нулю).
Основной проблемой является то, что сходимость медленная: оценка содержит коэффициент вида $\beta^{\frac{1}{k + 1}}$ ($\beta$~--- один из параметров).

В данной статье мы покажем, что с асимптотической точки зрения предыдущая оценка неулучшаема, однако, если дополнительно предположить, что мера $k$"~антиклик отделена от нуля, то можно значительно улучшить оценку.
Техника доказательства будет аналогична предложенной в~\cite{push17}.
Сперва мы получим искомый результат для конечного полуметрического пространства с равномерной мерой.
Вместо максимальной кластерной структуры будем оценивать жадную кластерную структуру, построение которой вполне конструктивно (хотя апеллирует к поиску максимальной клики в графе). Завершающим шагом будет обобщение на случай произвольного компактного пространства с помощью теоремы Бляшке~\cite{polovin2004}.

\section{Постановка задачи}

Пусть дано компактное метрическое пространство $(X, \rho)$ с ограниченной борелевской мерой $\mu$.
Любое борелевское подмножество~$X$ диаметра не более~$r$ будем называть~\textit{$r$-кластером}.

\begin{dfn}
Семейство $2r$"~кластеров $\mathcal{X} = \{X_1, \dots X_k\}$ будем называть $r$"~кластерной структурой порядка~$k$, если $\rho(X_i, X_j) \gs r$ при всех~$1 \ls i < j \ls k$, где~$\rho(A, B) = \inf \{\rho(x, y)\colon x \in A,\, y \in B\}$.
Мерой~$\mathcal{X}$ назовем величину $\mu(\mathcal{X}) \stackrel{\textrm{def}}{=} \displaystyle\sum_{i = 1}^k \mu(X_i)$.
\label{def:cluster_structure}
\end{dfn}

Пару точек $(x, y) \in X^2$ будем называть \textit{ребром}, длина ребра~--- это~$\rho(x,y)$.
Если $\rho(x, y) \ls r$, то будем называть ребро~$(x, y)$ \textit{$r$"~коротким}; если $\rho(x, y) > 3r$, то будем называть ребро~$(x, y)$ \textit{$r$"~длинным}; все остальные ребра~--- \textit{$r$"~средние}.
Набор точек $(x_1, \dots, x_{k})$ назовем \textit{$r$"~антикликой порядка $k$}, если $\rho(x_i, x_j) > r$ при всех $1 \ls i < j \ls k$.
Если понятно, о каком~$r$ идет речь, то приставка~$r$ будет опускаться.

В предыдущей статье рассматривались следующие параметрические ограничения на меру средних ребер и антиклик порядка $k + 1$:
\begin{equation}
M(X) = \frac{1}{2}\mu\{(x, y) \in X^2 \colon r < \rho(x, y) \ls 3r\} \ls \frac{1}{2}\delta \mu(X)^2,
\label{eq:med_dist}
\end{equation}
\begin{equation}
T_{k + 1}(X) = \frac{1}{(k + 1)!}\mu\{(x_1, \dots x_{k + 1}) \in X^{k + 1} \colon \rho(x_i, x_j) > r,\, 1 \ls i < j \ls k + 1\} \ls \frac{\beta \mu(X)^{k + 1}}{(k + 1)!},
\label{eq:k+1_anticlique}
\end{equation}
где $\delta, \beta > 0$~--- параметры. Была получена следующая оценка на меру $r$"~кластерной структурой порядка~$k$ максимальной меры $\mathcal{X}^*$:
\begin{equation}
\mu(\mathcal{X}^*) \gs \mu(X)(1 - \sqrt{\delta}(2k + 1) - (k(e + 1) + 1)\beta^{\frac{1}{k + 1}})
\label{eq:old_result}
\end{equation}

Следующее утверждение показывает, что по параметру $\beta$ с асимптотической точки зрения данная оценка принципиально улучшена быть не может.
\begin{state}
Пусть фиксированы $r > 0$ и $0 < \beta < 1$. Существует конечное метрическое пространство $(X, \rho)$ с равномерной мерой такое, что $M(X) = 0$, $T_{k + 1}(X) \ls \dfrac{1}{(k + 1)!}\beta |X|^{k + 1}$ и $|X| - \mu(\mathcal{X}^*) \gs \dfrac{1}{k + 1}\beta^{\frac{1}{k}}$.
\label{state:tight_bound}
\end{state}
\begin{proof}
Пусть $X = B_0 \sqcup B_1 \sqcup \ldots \sqcup B_{k}$, $|B_j| = \lambda|X|$ при $1 \ls j \ls k$, и $\lambda(k + 1) \ls 1$.
$$
\rho(x, y) = \left\{
  \begin{array}{ll}
    r, & x, y \in B_j \\
    4r, & x \in B_i, y \in B_j,\, i \neq j
  \end{array}
\right.
$$

Тогда $T_{k + 1}(X) = |X|^{k + 1}(1 - k\lambda) \lambda^k$ и $|X| - \mu(\mathcal{X}^*) = \lambda|X|$. Если взять $\lambda = \dfrac{\beta^{\frac{1}{k}}}{k + 1}$, то получаем требуемое утверждение.

\end{proof}

В вышеописанной конструкции при малых $\lambda$ фактически присутствует один кластер, а не $k$.
Иными словами, меры кластеров в кластерной структуре могут сильно различаться.

Предлагается ввести дополнительной ограничение, которое будет \textit{балансировать} меры кластеров.
Мы потребуем, чтобы мера антиклик порядка $k$ была ограничена снизу в следующем смысле
\begin{equation}
T_{k}(X) = \frac{1}{k!}\mu\{(x_1, \dots x_{k}) \in X^{k} \colon \rho(x_i, x_j) > r,\, 1 \ls i < j \ls k\} \gs \frac{\alpha \mu(X)^{k}}{k!}.
\label{eq:k_anticlique}
\end{equation}

Далее будет показано, как, используя условие~(\ref{eq:k_anticlique}), можно улучшить оценку на~$\mu(\mathcal{X}^*)$.
Как и ранее, мы будем рассматривать конечное полуметрическое пространство с равномерной мерой, а затем обобщим оценку на случай произвольного компактного пространства.
Вместо кластерной структуры максимальной меры рассмотрим жадную кластерную структуру.

\section{Жадная кластерная структура}

Пусть $X_1$~--- множество максимальной мощности среди всех $2r$"~кластеров (если таких множеств несколько, то выберем любое). Обозначим его окрестность $r$ за $Z_1$, т.е.
$$
Z_1 = \{x \in X \colon \rho(x, X_1) < r\}
$$
Пусть у нас есть попарно непересекающиеся множества $Z_1, \dots, Z_m$. Тогда $X_{m + 1}$~--- множество максимальной мощности среди всех $2r$"~кластеров в $X \setminus \displaystyle\bigcup_{i = 1}^{m} Z_i$, а множество $Z_{m + 1}$~--- $r$"~окрестность $X_{m + 1}$ во множестве $X \setminus \displaystyle\bigcup_{i = 1}^{m} Z_i$, т.е.
$$
Z_{m + 1} = \left\{x \in X \setminus \bigcup_{i = 1}^{m} Z_i \colon \rho(x, X_{m + 1}) < r\right\}
$$
Так как мощность $X$ конечна, то процедура оборвётся не некотором шаге.

\begin{dfn}
Построенное разбиение $X = \displaystyle\bigsqcup_{i = 1}^n Z_i$ мы назовем жадным кластерным разбиением, а семейство $2r$"~кластеров $\{X_1,\dots X_k\}$ назовем жадной $r$"~кластерной структурой порядка~$k$.
\end{dfn}

Рассмотрим множества~$Z_i$ и~$X_i \subset Z_i$.
Для любых $x \in X_i$ и $z \in Z_i$ выполнено $\rho(x, z) \ls 3r$, поэтому концы всех длинных ребер лежат в $Z_i \setminus X_i$.
Рассмотрим во множестве $Z_i \setminus X_i$ максимальное паросочетание из длинных ребер, которое покрывает множество $U_i$.
Пусть $Y_i = Z_i \setminus (X_i \cup U_i)$, тогда $X_i \cup Y_i$ является $3r$"~кластером.
\begin{state}
Пусть $(A, \rho)$~--- конечное полуметрическое пространство диаметра не более~$3r$, и множество~$B$ является $2r$"~кластером максимальной мощности. Тогда число средних ребер не менее $M(A) \gs \frac{1}{2}\max\{|A|, 2|B|\}|A\setminus B|$.
\label{sate:max_cluster}
\end{state}

Доказательство этого утверждения можно найти в~\cite{push14}

Также для любого ребра $(u_1, u_2)$ из паросочетания, покрывающего $U_i$, и точки $x \in X_i$ хотя бы одно из ребер $(x, u_j)$ является средним.
В купе с утверждением~\ref{sate:max_cluster} получаем следующее неравенства
\begin{equation}
M(Z_i) \gs \frac{1}{2}(|X_i| + |Y_i|)|Y_i| + \frac{1}{2}|U_i||X_i|
\label{eq:med_int}
\end{equation}
\begin{equation}
M(Z_i) \gs |X_i||Y_i| + \frac{1}{2}|U_i||X_i|
\end{equation}

Пусть $L(Z_i)$~--- количество длинных ребер во множестве $Z_i$, тогда
$$
L(Z_i) \ls |U_i||Y_i| + \dfrac{1}{2}|U_i|^2 \ls \frac{|U_i|}{|X_i|}M(Z_i)
$$

Пусть $I_1$~--- множество индексов $1 \ls i \ls n$ таких, что $|X_i|(k + 1) \ls |Z_i|$.

Если $i \notin I_1$, то $L(Z_i) \ls kM(Z_i)$ и $\displaystyle\sum_{i \notin I_1}L(Z_i) \ls kM(X) \ls \frac{1}{2}k\delta|X|^2$.

\begin{state}
$$
\sum_{i \in I_1}|Z_i| \ls \frac{(k + 1)\beta}{\alpha}|X|
$$
\label{sate:bad_zs}
\end{state}
\begin{proof}
Рассмотрим произвольную антиклику порядка $k$ $(a_1, \ldots, a_k)$. Пусть $B_r(a_j)$~--- замкнутый шар радиуса $r$ с центром $a_j$. Тогда $|B_r(a_j) \cap Z_j| \ls |X_j|$, и при $i \in I_1$ имеем $\left|Z_i \setminus \displaystyle\bigcup_{j = 1}^k B_r(a_j)\right| \geqslant \dfrac{1}{k + 1}|Z_i|$. Таким образом, для любой точки $b \in Z_i \setminus \displaystyle\bigcup_{j = 1}^k B_r(a_j)$ $(b, a_1, \ldots, a_k)$~--- антиклика порядка $k + 1$. Так как каждую антиклику порядка $k + 1$ можно получить подобным построением не более $k + 1$ раза, то, используя неравенства~\ref{eq:k+1_anticlique} и \ref{eq:k_anticlique}, получаем
$$
\dfrac{\beta |X|^{k + 1}}{(k + 1)!} \gs T_{k + 1}(X) \gs \dfrac{1}{(k + 1)^2}\sum_{i \in I_1}|Z_i| T_{k}(X) \gs \dfrac{\alpha |X|^k}{k!(k + 1)^2}\sum_{i \in I_1}|Z_i|
$$
\end{proof}

\section{Оценки числа антиклик}

Пусть $\{W_i\}_{i = 1}^n$ упорядоченные по убыванию мощности множеств $Z_i$.
Введем обозначение для симметрического многочлена от $n$ переменных
\begin{equation*}
\sigma_s(y_1,\dots y_n) \stackrel{\textrm{def}}{=} \sum_{1 \ls i_1 < \ldots < i_s  \ls n}\prod_{j = 1}^{s}y_j,
\label{eq:symmetric_poly}
\end{equation*}

Мы имеем следующую оценку снизу на число антиклик порядка $k + 1$, доказанную~в~\cite{push17}.
\begin{state}
$$
T_{k + 1}(X) \gs \frac{1}{(k + 1)!}\sigma_{k + 1}(W_1,\ldots, W_n)
$$
\label{sate:k+1_anticliqe_bound}
\end{state}

Далее мы получим верхнюю оценку числа антиклик порядка $k$ в терминах симметрических многочленов от $W_i$.

\begin{state}
Пусть $\delta + \dfrac{(k + 1)\beta^2}{\alpha^2} \ls \dfrac{2}{(k + 1)^3}$, тогда
$$
T_{k}(X) \ls \sigma_{k}(W_1,\ldots, W_n) + \frac{k\lambda|X|^k}{2(k - 2)!},
$$
где $\lambda = \dfrac{k + 1}{2}\delta  + \dfrac{(k + 1)^2\beta^2}{2\alpha^2}$
\label{sate:k1_anticliqe_bound}
\end{state}

\begin{proof}
Рассмотрим произвольную антиклику порядка $k$ $(a_1, \ldots, a_k)$. Проделаем с ней следующую процедуру. Изначально все точки $a_i$ не отмечены. Пока существуют $i$ и $j$ такие, что $a_i$ и $a_j$ лежат в одном $Z_l$ и не отмечены, мы отмечаем $a_i$ и $a_j$ вместе с ребром $(a_i, a_j)$. В конце процедуры мы получим набор $a_i$, попарно лежащих в разных $Z_l$ и набор средних или длинных ребер из множества
$$
\Lambda = \bigcup_{i = 1}^n\{(x, y) \in Z_i^2 \colon \rho(x, y) > r\},
$$

где пары $(x, y)$ считаются неупорядоченными.
Таким образом каждую, антиклику порядка $k$ можно закодировать $s$ ребрами из $\Lambda$ и $k - 2s$ токами из попарно различных $Z_i$, и
$$
T_{k}(X) \ls \sum_{s} \sigma_{k - 2s}(W_1,\ldots, W_n)|\Lambda|^s
$$
$$
|\Lambda| = \sum_{i = 1}^n (L(Z_i) + M(Z_i)) = \sum_{i \in I_1} (L(Z_i) + M(Z_i)) + \sum_{i \notin I_1} (L(Z_i) + M(Z_i)) \ls
$$
$$
\ls \frac{1}{2}\sum_{i \in I_1} W_i^2 + (k + 1)\sum_{i \notin I_1} M(Z_i) \ls \frac{1}{2}\left(\sum_{i \in I_1} W_i\right)^2 + (k + 1)M(X)
$$
Используя утверждение~\ref{sate:bad_zs} и неравенство~(\ref{eq:med_dist}), получаем
$$
|\Lambda| \ls \frac{k + 1}{2}\delta |X|^2 + \frac{(k + 1)^2\beta^2}{2\alpha^2} |X|^2 \stackrel{\textrm{def}}{=} \lambda |X|^2
$$
Так как $\lambda \ls \dfrac{1}{(k + 1)^2}$ и
$$
\sigma_{s}(W_1,\ldots, W_n) \ls {n \choose s}\frac{|X|^s}{n^s} \ls \frac{|X|^s}{s!},
$$
то
$$
T_{k}(X) \ls \sigma_{k}(W_1,\ldots, W_n) + \sum_{s \geq 1} \frac{|X|^{k - 2s}}{(k - 2s)!}\lambda^s |X|^{2s} \ls \sigma_{k}(W_1,\ldots, W_n) + \frac{k}{2}\frac{\lambda|X|^k}{(k - 2)!}
$$
\end{proof}

Далее мы будем предполагать, что условия утверждения~\ref{sate:k1_anticliqe_bound} выполнены. Используя неравенства~(\ref{eq:k_anticlique}) и~(\ref{eq:k+1_anticlique}), а также утверждения~\ref{sate:k+1_anticliqe_bound} и~\ref{sate:k1_anticliqe_bound} получаем
\begin{equation}
\sigma_{k + 1}(W_1,\ldots, W_n) \ls \beta |X|^{k + 1}
\end{equation}
\begin{equation}
\sigma_{k}(W_1,\ldots, W_n) \gs \frac{1}{k!}\left(\alpha - \frac{1}{2}\lambda k^3\right)|X|^{k} \stackrel{\textrm{def}}{=} \frac{1}{k!}\alpha' |X|^k
\end{equation}

Теперь мы можем получить простую оценку на $\displaystyle\sum_{i = 1}^kW_i$
\begin{state}
$$
\sum_{i = 1}^kW_i \gs \left(1 - \frac{(k + 1)!\beta}{\alpha'}\right)|X|
$$
\label{sate:frist_k}
\end{state}
\begin{proof}
Пусть $\mathcal{J}$~--- семейство всех $k$"~элементных подмножеств множества $N = \{1, 2,\ldots, n\}$. Так как последовательность $\{W_i\}_{i = 1}^n$ монотонно убывает, то
$$
\sum_{i = k + 1}^nW_i \sigma_{k}(W_1,\ldots, W_n) = \sum_{J \in \mathcal{J}} \prod_{j \in J} \left(W_j\sum_{i = k + 1}^nW_i \right) \ls
$$
$$
\ls \sum_{J \in \mathcal{J}} \prod_{j \in J} \left(W_j\sum_{i \in N \setminus J}W_i\right) = (k + 1)\sigma_{k + 1}(W_1,\ldots, W_n)
$$
$$
\sum_{i = 1}^kW_i = |X| - \sum_{i = k + 1}^nW_i \gs \frac{(k + 1)\sigma_{k + 1}(W_1,\ldots, W_n)}{\sigma_{k}(W_1,\ldots, W_n)} \gs \left(1 - \frac{(k + 1)!\beta}{\alpha'}\right)|X|
$$

\end{proof}

\section{Оценка мощности кластерной структуры}

Пусть $I_0$~--- множество индексов, соответствующих $k$ наибольшим по мощности множествам $Z_i$. Нам осталось оценить величину $\displaystyle\sum_{i \in I_0} |Z_i| - \displaystyle\sum_{i \in I_0} |X_i|$.

Если $i \in I_0 \cap I_1$, то по утверждению~\ref{sate:bad_zs}
$$
\sum_{i \in I_0 \cap I_1} |Z_i| \ls \frac{(k + 1)\beta}{\alpha}|X|
$$

Если $i \in I_0 \setminus I_1$, то $(k + 1)|X_i| \gs |Z_i|$ и, используя неравенство~(\ref{eq:med_int}), получаем
$$
M(Z_i) \gs \frac{1}{2(k + 1)}|Z_i|(|Y_i| + |W_i|) = \frac{1}{2(k + 1)}|Z_i|(|Z_i| - |X_i|)
$$
$$
|Z_i| - |X_i| \ls \frac{2(k + 1)M(Z_i)}{|Z_i|}
$$
Пусть $I_2$~--- множество индексов таких, что $|Z_i| \gs \sqrt{\delta}|X|$, тогда
$$
\sum_{i \in I_0 \setminus I_1} (|Z_i| - |X_i|) \ls \sum_{i \in I_0 \cap I_2 \setminus I_1} (|Z_i| - |X_i|) + k\sqrt{\delta}|X| \ls
$$
$$
 \ls k\sqrt{\delta}|X| + \frac{2(k + 1)M(X)}{\sqrt{\delta}|X|} \ls (2k + 1)\sqrt{\delta}|X|
$$
Наконец,
$$
\sum_{i \in I_0} |Z_i| - \sum_{i \in I_0} |X_i| \ls (2k + 1)\sqrt{\delta}|X| + \frac{(k + 1)\beta}{\alpha}|X|
$$
В сочетании с утверждением~\ref{sate:frist_k} имеем
$$
|X| - \sum_{i \in I_0} |X_i| \ls (2k + 1)\sqrt{\delta}|X| + \frac{(k + 1)\beta}{\alpha}|X| + \frac{(k + 1)!\beta}{\alpha'}|X| \ls
$$
$$
\ls (2k + 1)\sqrt{\delta}|X| + \frac{k!(k + 2)\beta}{\alpha - \frac{1}{2}k^3\lambda}|X|
$$
Так как система множеств $\{X_i\}_{i \in I_0}$ является $r$"~кластерной структурой порядка $k$, то мы получаем следующий результат.
\begin{thm}
Пусть $(X, \rho)$ конечное полуметрическое пространство с равномерной мерой $\mu$, а $\mathcal{X}^*$~--- $r$"~кластерная структура максимальной меры. Тогда, если выполнены неравенства~(\ref{eq:med_dist}),~(\ref{eq:k+1_anticlique}),~(\ref{eq:k_anticlique}) и $\delta + \dfrac{(k + 1)\beta^2}{\alpha^2} \ls \dfrac{2}{(k + 1)^3}$, то
\begin{equation}
\mu(\mathcal{X}^*) \gs \Psi(\alpha, \beta, \delta)|X|,
\label{eq:main_th}
\end{equation}
где
$$
\Psi(\alpha, \beta, \delta) = 1 - \sqrt{\delta}(2k + 1) - \frac{k!(k + 2)\beta}{\alpha - \frac{1}{2}k^3\lambda}
$$
$$
\lambda = \dfrac{k + 1}{2}\delta  + \dfrac{(k + 1)^2\beta^2}{2\alpha^2}
$$
\label{th:finite}
\end{thm}

\section{Обобщение на случай произвольного компактного пространства}

Нам осталось обобщить предыдущую теорему на случай произвольного компактного пространства.
\begin{thm}
Пусть $(X, \rho)$ компактное метрическое пространство с ограниченной борелевской мерой $\mu$, $\mathcal{X}^*$~--- $r$"~кластерная структура максимальной меры, и функция распределения величины $\rho(x, y)$ непрерывна. Тогда, если выполнены неравенства~(\ref{eq:med_dist}),~(\ref{eq:k+1_anticlique}),~(\ref{eq:k_anticlique}) и $\delta + \dfrac{(k + 1)\beta^2}{\alpha^2} \ls \dfrac{2}{(k + 1)^3}$, то выполнено неравенство~(\ref{eq:main_th}).
\end{thm}

\begin{proof}

Доказательство почти дословно повторяет доказательство аналогичного результата из~\cite{push17}.

Фиксируем произвольное $0 < \eps < r$. В $X$ существует конечная~$\eps$"~сеть, а значит и разбиение~$X$ на конечное число~$N_{\eps}$ $\eps$"~кластеров $\{A_i\}_{i = 1}^{N_{\eps}}$.
Выберем $N_{\eps}$ положительных рациональных чисел $q_1, \dots q_{N_{\eps}}$ так, что $\mu(A_i) \gs q_i$ при $1 \ls i \ls N_{\eps}$ и $q_i \gs \mu(A_i)(1 - \eps)$.

Рассмотрим полуметрическое пространство конечной мощности $X_{\eps} = B_1 \sqcup \ldots \sqcup B_s$, где $\dfrac{|B_i|}{|B_j|} = \dfrac{q_i}{q_j}$, а функция расстояния $\rho_{\eps}$ определяется следующим образом:
$$
\rho_{\eps}(x, y) = \left\{
  \begin{array}{ll}
    0, & x, y \in B_i \\
    \rho(A_i, A_j), & x \in B_i, y \in B_j,\, i \neq j
  \end{array}
\right.
$$
Отметим, что
$$
\frac{(1 - \eps)\mu(A_i) |X_{\eps}|}{\mu(X)} \ls |B_i| = \frac{q_i |X_{\eps}|}{\displaystyle\sum_{j = 1}^{N_{\eps}}q_j} \ls \frac{\mu(A_i) |X_{\eps}|}{(1 - \eps)\displaystyle\sum_{j = 1}^{N_{\eps}}\mu(A_j)} = \frac{\mu(A_i) |X_{\eps}|}{(1 - \eps)\mu(X)}
$$

Чтобы завершить доказательство достаточно показать, что
$$
T_{k}(X_{\eps}) \gs \frac{1}{k!}\alpha_{\eps}|X_{\eps}|^{k},
$$
где $\alpha_{\eps} \to \alpha$ при $\eps \to 0$.
$$
T_{k}(X_{\eps}) = \sum_{1 \ls i_1 < \dots < i_{k} \ls N_{\eps}}\prod_{1 \ls j < l \ls k} [\rho(A_{i_j}, A_{i_l}) > r] \prod_{j = 1}^{k} |B_{i_j}| \gs
$$
$$
\gs \frac{(1 - \eps)^{k}|X_{\eps}|^{k}}{\mu(X)^{k}}\sum_{1 \ls i_1 < \dots < i_{k} \ls N_{\eps}}\prod_{1 \ls j < l \ls k} [\rho(A_{i_j}, A_{i_l}) > r] \prod_{j = 1}^{k} \mu(A_{i_j}) \gs \frac{(1 - \eps)^{k}|X_{\eps}|^{k}}{\mu(X)^{k}}T^{\eps}_{k}(X),
$$
где $T^{\eps}(X)$~--- мера $r - 2\eps$"~антиклик порядка $k$ в $X$.
Рассмотрим множество $r - 2\eps$"~антиклик порядка $k$ в $X$, не являющихся $r$"~антикликами
$$
\mathcal{T} = \{(x_1, \ldots, x_k) \in X^{k}\colon \rho(x_i, x_j) > r - 2\eps,\, 1 \ls i < j \ls k + 1, \exists l, s\; \rho(x_l, x_s) \ls r\}
$$
Заметим, что в силу непрерывности функции распределения $\rho(x, y)$ при $\eps \to 0$
$$
\mu(\mathcal{T}) \ls \frac{k^2}{2}\mu(X)^{k - 2}\mu\{(x, y) \in X^2\colon r - 2\eps < \rho(x, y) \ls r\} \to 0.
$$
Тогда
$$
T^{\eps}_{k}(X) \gs T_k(X) - \frac{1}{k!}\mu(\mathcal{T}) \gs \frac{1}{k!}\alpha\mu(X)^k - \frac{k^2}{2k!}\mu(X)^{k - 2}\mu\{(x, y) \in X^2\colon r - 2\eps < \rho(x, y) \ls r\}
$$
$$
T_{k}(X_{\eps}) \gs \frac{1}{k!}|X_{\eps}|^{k}\left(\alpha - \frac{k^2}{2\mu(X)^2}\mu\{(x, y) \in X^2\colon r - 2\eps < \rho(x, y) \ls r\}\right)
$$

\end{proof}


\end{document}